\documentclass[a4paper,english]{lipics-v2018}
\usepackage{cite,microtype,graphicx}

\nolinenumbers

\title{The Parameterized Complexity of Finding Point Sets with Hereditary Properties}
\author{David Eppstein}{Computer Science Department, University of California, Irvine, USA}{eppstein@uci.edu}{}{Supported in part by NSF grants  CCF-1618301 and CCF-1616248.}
\author{Daniel Lokshtanov}{Department of Informatics, University of Bergen, Norway}{dlokshtanov@gmail.com}{}{Supported by Pareto-Optimal Parameterized Algorithms, ERC Starting Grant 715744}
\authorrunning{D. Eppstein and D. Lokshtanov}
\Copyright{David Eppstein and Daniel Lokshtanov}
\subjclass{}
\hideLIPIcs

\usepackage{aliascnt}
\makeatletter
\let\lemma\@undefined
\let\endlemma\@undefined
\makeatother
\theoremstyle{theorem}
\newaliascnt{lemma}{theorem}
\newtheorem{lemma}[lemma]{Lemma}
\aliascntresetthe{lemma}

\newcommand{\Wone}{\mathrm{W}[1]}

\begin{document}
\maketitle

\begin{abstract}
We consider problems where the input is a set of points in the plane and an integer $k$, and the task is to find a subset $S$ of the input points of size $k$ such that $S$ satisfies some property. We focus on properties that depend only on the order type of the points and are monotone under point removals. We show that not all such problems are fixed-parameter tractable parameterized by $k$, by exhibiting a property defined by three forbidden patterns for which finding a $k$-point subset with the property is $\Wone$-complete and
(assuming the exponential time hypothesis) cannot be solved in time $n^{o(k/\log k)}$. However, we show that problems of this type are fixed-parameter tractable for all properties that include all collinear point sets, properties that exclude at least one convex polygon, and properties defined by a single forbidden pattern.
\end{abstract}

\section{Introduction}
In this work, we study the parameterized complexity of finding subsets of planar point sets having hereditary properties, by analogy to past work on hereditary properties of graphs.

In graph theory, a hereditary properties of graphs is a property closed under induced subgraphs. These have long been a central organizing principle for such diverse topics as coloring, perfection, intersection graph theory, and the theories of claw-free and triangle-free graphs. Algorithmically, finding the largest induced subgraph with a given hereditary property is NP-hard for any nontrivial property. The decision version of this problem is NP-complete whenever in addition the property can be tested in polynomial time. Following the proof of this result by Lewis and Yannakakis, this became one of the fundamental problems listed as hard in Garey and Johnson's early survey of the theory of NP-completeness~\cite{GT21,LewYan-JCSS-80}. Hereditary properties generalize closure under subgraphs or under graph minors, and finding induced subgraphs from families with these stronger closure properties has also been studied. For instance, in graph drawing, it is of interest to find large planar induced subgraphs as a step towards drawing nonplanar graphs~\cite{BorEppZhu-JGAA-15}.

The parameterized complexity of finding $k$-vertex induced subgraphs with hereditary properties (parameterized by $k$) has been completely resolved by Khot and Raman~\cite{KhoRam-TCS-02}. If all cliques and all independent sets have the given property, then (by Ramsey's theorem) every $n$-vertex graph has an $\Omega(\log n)$-vertex induced subgraph with the property. Thus, the only graphs that might not have an induced subgraph of size $k$ with the property are those of bounded size, making the problem fixed-parameter tractable. If there exist a clique and an independent set that both do not have the property, then by the same reasoning all graphs with the property have bounded size, and finding the largest of these that exists within a given input is fixed-parameter tractable. In all remaining cases, as Khot and Raman show, the problem is $\Wone$-complete. For hereditary properties that can be characterized by finitely many forbidden induced subgraphs, the dual problem of removing $k$ vertices so that the remaining induced subgraph has the given property is always fixed-parameter tractable, as a hitting set problem in which we must choose $k$ vertices to hit all copies of the obstacles~\cite{PK,Abu-JCSS-10}. The parameterized complexity of problems of this type that do not have a finite number of obstacles, such as deletion to a planar graph ($k$-apex graph recognition)~\cite{Kaw-FOCS-09} or to a bipartite graph (odd cycle transversal)~\cite{KawRee-SODA-10,KriNar-IPL-13}, also remains of interest.

Point sets in the Euclidean plane can be described combinatorially by their \emph{order type}, the specification for each ordered triple of points of whether that triple forms the vertices of a triangle in clockwise order or in counterclockwise order, or whether the three points are collinear.
From this point of view, it is natural to define a hereditary property of point sets to be a Boolean function of point sets that depends only on the order type, and that remains true for every subset of a point set having the property. Although this concept of hereditary properties has only been adopted very recently as a central organizing principle in discrete geometry~\cite{Epp-18}, it encompasses many famous and well-studied problems. For instance:
\begin{itemize}
\item The \emph{happy ending problem} concerns conditions that force some $k$-point subset to be in convex position~\cite{ErdSze-CM-35,Suk-JAMS-17}.  The property of not containing a $k$-point convex subset is hereditary~\cite[Chapter 11]{Epp-18}. It can be tested in polynomial time for points in the plane~\cite{ChvKli-SEC-80,EdeGui-JCSS-89}.
However, it is NP-hard in three dimensions, and a related problem of finding the largest 3d $k$-point convex set that does not enclose any other input points is $\Wone$-hard~\cite{GiaKnaWer-ESA-13}.
\item The \emph{no-three-in-line problem} concerns the size of the largest general-position subset of an $n\times n$ grid~\cite{Rot-JLMS-51,HalJacSud-JCTA-75}. Here, ``general position'' means having no three collinear points.
The property of not containing a $k$-point general-position subset is hereditary~\cite[Chapter 9]{Epp-18}.
\item The \emph{Erd{\H{o}}s--Ulam problem} and \emph{Harborth's conjecture} both concern realizability with rational distances. The Erd{\H{o}}s--Ulam problem asks whether there exist dense sets with all distances rational~\cite{Tao-14}, while Harborth's conjecture is that every planar graph can be drawn with non-crossing straight edges of rational length~\cite{HarKemMol-EM-87}. The existence of a rational-distance point set with the same order type as the given point set is hereditary. If every general-position order type has such a realization then Harborth's conjecture follows, and if not then the Erd{\H{o}}s--Ulam problem has a negative answer~\cite[Section 13.5]{Epp-18}.
\end{itemize}

Every hereditary property can be defined by its \emph{forbidden patterns}, the point-minimal order types that do not have the property. The first author's recent book on these problems asked whether every problem
of finding $k$ points with a hereditary property, out of a larger finite set of points in the plane, 
for a property defined by finitely many forbidden patterns, is fixed-parameter tractable~\cite[Open Problem 7.6]{Epp-18}. In this paper, we answer this question negatively, by finding a property for which (under standard complexity-theoretic assumptions) no FPT algorithm exists. However, we show that several natural classes of properties have fixed-parameter tractable algorithms. Specifically, we show the following results:
\begin{itemize}
\item There exists a hereditary property $\Pi$, defined by finitely many forbidden patterns, such that finding a $k$-point subset with property $\Pi$ is $\Wone$-complete and (assuming the exponential time hypothesis)  cannot be solved in time $n^{o(k/\log k)})$.
\item Finding a $k$-point subset with a given hereditary property is fixed-parameter tractable whenever all collinear sets have the property.
\item Finding a $k$-point subset with a given hereditary property is fixed-parameter tractable whenever there exists a convex polygon that does not have the property.
\item Finding a $k$-point subset that avoids a single forbidden pattern is always fixed-parameter tractable.
\end{itemize}
As with the analogous graph problems, we observe that the dual problem of removing $k$ points to make the remaining points have a given property is fixed-parameter tractable whenever the property has finitely many forbidden patterns, by treating it as a hitting set problem for the copies of the patterns within the point set~\cite[Theorem 7.9]{Epp-18}.

Our algorithmic results apply to any model of point set input that allows us to determine the orientation of a triple of points in constant (or polynomial) time. For points given by Cartesian coordinates, these orientations can be determined from the sign of a quadratic polynomial of these coordinates. Not all order types can be realized with integer coordinates, and even when this is possible a realization may need very large coordinates~\cite{GooPolStu-STOC-89}, but our algorithms could instead use inputs that specify an $n\times n\times n$ three-dimensional array of orientations.
The point sets used for our hardness results and lower bounds can be constructed using explicit integer coordinates of polynomial magnitude.

\section{A hard property}


\begin{figure}[t]
\centering
\includegraphics[scale=0.4]{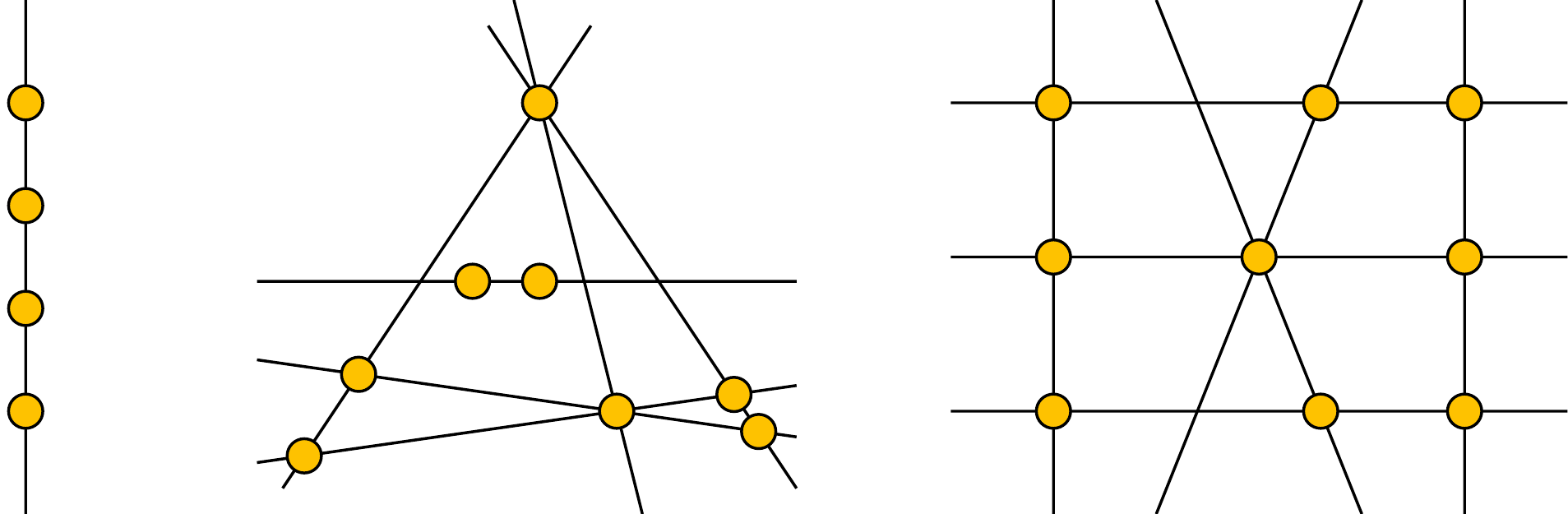}
\caption{The three forbidden patterns for a compliant point set.}
\label{fig:forbidden}
\end{figure}

The property for which we will prove finding a $k$-point subset hard is defined by the three forbidden patterns depicted in \autoref{fig:forbidden}: four points in a line (left), an eight-point pattern consisting of the six points and four three-point lines of a complete quadrilateral,\footnote{A \emph{complete quadrilateral} consists of four lines and six points, with each line containing three  points and each point contained in two lines. It can be constructed from four Euclidean lines, no two parallel and no three crossing in a single point, by including the six crossing points of pairs of lines. By using lines whose equations have rational coefficients and then clearing denominators from the point coordinates, it can be constructed using points with integer coordinates.} with two additional points enclosed by it (center), and a nine-point pattern resembling the three rows and columns of a tic-tac-toe board or $3\times 3$ grid, but with the central vertical line broken rather than straight (right). We call a set of points \emph{compliant} if it has no subset of points whose order type is the same as one of these three patterns.

The hard inputs to our problem will be point sets of a special form, which we call \emph{yards}.
A yard consists of points with the following properties:
\begin{itemize}
\item Most of its points are arranged onto horizontal lines of at least three points per line, which we call \emph{rows}. We call the leftmost and rightmost points of each row the \emph{guards} of the row (or more specifically the left and right guards, respectively). We call the remaining points, that are neither leftmost nor rightmost in their row, the \emph{inmates}.
\item Some triples of points from different rows form non-horizontal lines of three points, but there are no non-horizontal lines of four points.
\item Every line between two inmates of different rows passes completely above or completely below all left guards, and completely below or completely above (respectively) all right guards. Similarly, every line through two left guards or through two right guards passes completely above or completely below all inmates. (These points may determine vertical lines, which we consider to pass above and below all points.)
\item The only lines that contain three or more points and include both inmate and guard points are the rows.
\item No subset of nine guards forms the configuration on the right of \autoref{fig:forbidden}.
\item In addition to the points on rows, the yard has six additional points, forming a complete quadrilateral. We call these points the \emph{fence}.
\item The only lines of three or more points that include points of the fence are the four lines from which the fence was defined. No other six points of the yard form a complete quadrilateral.
\item The fence encloses the inmates but not the guards.  Every two inmates on a single row, plus the six points of the fence, form the middle forbidden pattern of \autoref{fig:forbidden}. Conversely, every instance of the middle forbidden pattern is formed by two inmates on a single row plus the fence.
\item No two inmates on different rows from each other form (with the fence) the middle configuration. That is, for each two inmates on different rows, the line through these two points passes through the fence in a combinatorially distinct way than the way that the rows pass through the fence.
\end{itemize}

An example of a yard is shown in \autoref{fig:yard}. For a yard with $r$ rows,
we define a \emph{lineup} to be a compliant subset of $k=3r+6$ points from the yard.

\begin{figure}[t]
\centering
\includegraphics[scale=0.4]{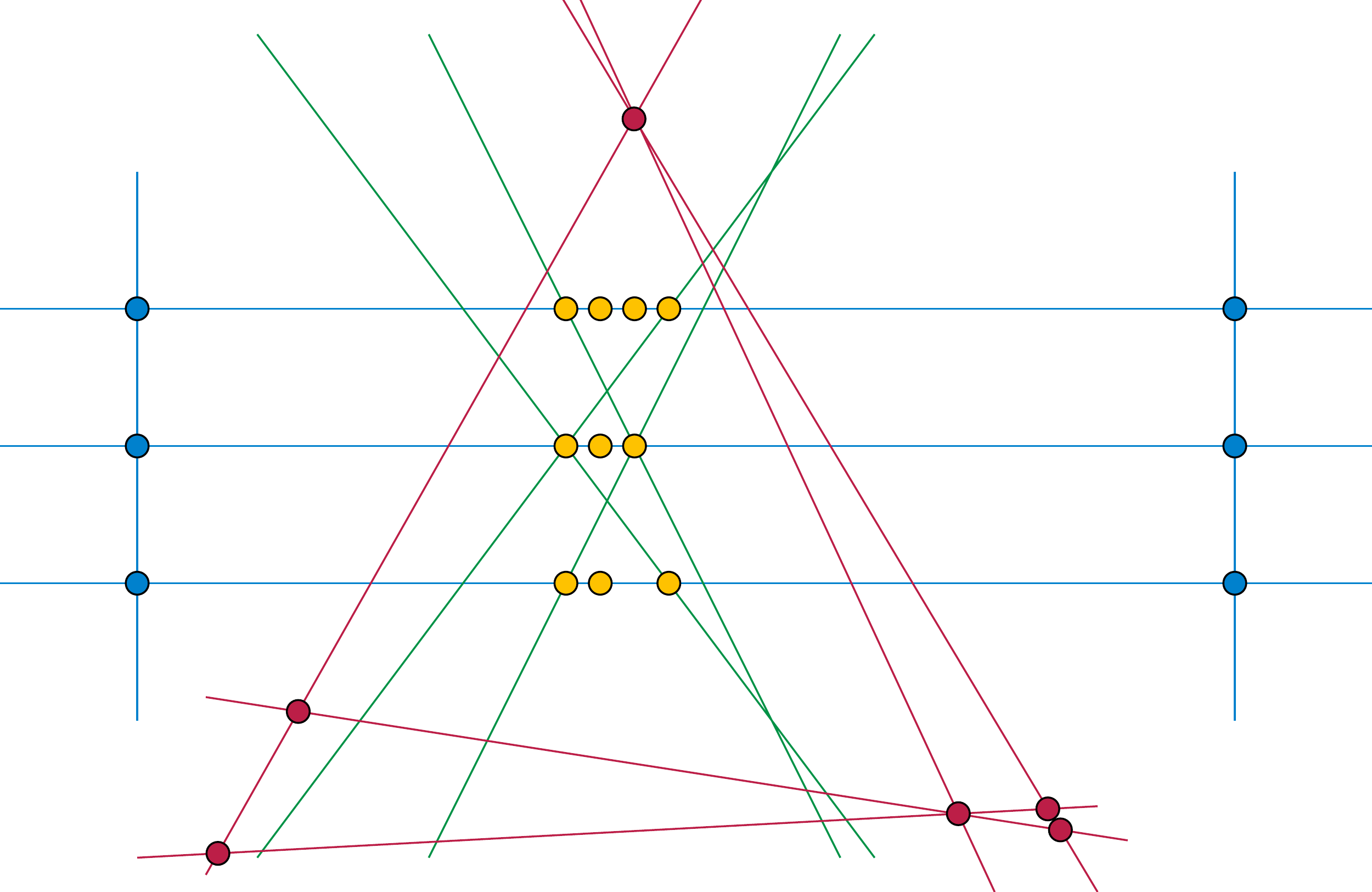}
\caption{A yard with three rows. The fence and its points are shown in red; the rows and their guards are blue, and the inmates on each row are yellow. The green lines are the most far from vertical among the lines through pairs of inmates that are not both on the same row as each other; they can be used to verify that all non-horizontal and non-vertical lines through pairs of inmates pass above or below the guards on each side, and that no pair of inmates on different rows from each other forms a forbidden pattern with the fence. By \autoref{lem:lineup}, every lineup in this yard consists of all of the blue and red points together with three collinear yellow points, one from each row. In this example all the left guards and all the right guards are collinear, but in larger examples only certain triples of guards would be collinear.}
\label{fig:yard}
\end{figure}

\begin{lemma}
\label{lem:lineup}
Every lineup must consist of all six points of the fence, and two guards and one inmate from each row. Whenever three rows have the property that their three left guards and their three right guards
form two collinear triples of points, the three inmates of the lineup from these rows must also be collinear.
\end{lemma}

\begin{proof}
The forbidden four-point line forces the lineup to contain at most three points from each row, so if it contains $3r+6$ points it must include exactly three points from each row and all six points of the fence.
The forbidden pattern in the form of a complete quadrilateral plus two points forces the lineup to contain at most one inmate from each row, so in order for it to contain three points from each row it must contain both guards.

If the three left guards and the three right guards of a triple of rows are collinear (as they are in \autoref{fig:yard}), then the three inmates of the lineup from the same three rows must also be collinear. For otherwise, they would form the forbidden pattern on the right of \autoref{fig:forbidden}, either as shown or under a $180^\circ$ rotation (which causes no change to the order type of the configuration). Therefore, every such triple of inmates must be collinear.
\end{proof}

\autoref{lem:lineup} is not a necessary and sufficient condition (for instance, it is possible to have a subset of points that meets the condition of the lemma but contains a copy of the rightmost forbidden pattern, rotated by $90^\circ$) but we will see that it is necessary and sufficient for the yards constructed by the reduction that we describe next.


We prove hardness of finding lineups in yards via a parameterized reduction from the problem of \emph{partitioned subgraph isomorphism}. This problem's input consists of two graphs, a large \emph{host graph} $H$ and a smaller pattern graph $G$, whose sizes (numbers of edges) are $n$ and $k$ respectively. The vertices of both $G$ and $H$ are colored, with all vertices of $G$ having distinct colors. The problem is to find a subgraph of $H$ that is isomorphic to $G$ and matches it in coloring. Thus, the color classes of $H$ partition its vertices into subsets. We must choose one representative vertex from each subset so that whenever two vertices in $G$ are adjacent their chosen representatives in $H$ are adjacent as well. 
The color partition of the vertices in $H$ is what gives rise to the name, partitioned subgraph isomorphism. In our hardness reduction, we will consider a restricted case of this problem, \emph{cubic partitioned subgraph isomorphism}, in which we require the pattern graph $G$ to be a cubic (that is, 3-regular) graph. We make no such restriction on $H$. This problem is $\Wone$-complete and, assuming the exponential time hypothesis, cannot be solved in time $n^{o(k/\log k)}$.

Cubic partitioned subgraph isomorphism has been used for the same sort of hardness reduction previously~\cite{AmiKreMar-MFCS-16,BonMilRza-Algo-17}.
The $\Wone$-completeness of partitioned subgraph isomorphism problem appears to be folklore. Subgraph isomorphism includes as a special case the (uncolored) clique problem, which is $\Wone$-complete. It can be reduced to the colored problem by taking the tensor product of the host graph with a clique and coloring the vertices of the product by which copy of the host graph they come from.\footnote{In more detail, the tensor product $G\times H$ of two graphs $G$ and $H$ has as its vertices the pairs $(u,v)$ where $u$ is a vertex in $G$ and $v$ is a vertex in $H$. Two pairs $(u,v)$ and $(u',v')$ are adjacent in the tensor product if and only if either $u=u'$ and $v$ is adjacent to $v'$, or $v=v'$ and $u$ is adjacent to $u'$. If we properly color $K_k$ and assign each vertex $(u,v)$ in $G\times K_k$ the same color as $v$, then $G$ has a $k$-clique if and only if the partitioned subgraph isomorphism instance for $G\times K_k$ and $K_k$ (with these colors) has a solution. This reduction proves the $\Wone$-completeness of partitioned subgraph isomorphism problem from the known completeness of the clique problem.} In turn, partitioned subgraph isomorphism on graphs of minimum degree at least three can be reduced to the 3-regular case by expanding each high-degree vertex by a $3$-regular tree, as follows:
\begin{itemize}
\item Choose for each vertex $v$ of the pattern graph a tree $T_v$ having degree three at each internal vertex and having the neighbors of $v$ as its leaves. Let $T'_v$ be the subtree of $T_v$ consisting only of the interior nodes of $T_v$.
\item Form the disjoint union of the trees $T'_v$. For each edge $uv$ of the pattern graph add an edge to this union, connecting the node in $T'_u$ to which the leaf for $v$ is attached to the node in $T'_v$ to which the leaf for $u$ is attached.
\item For each vertex $w$ of the host graph, having the same color as a vertex $v$ of the pattern graph, make a tree $T'_w$ isomorphic to $T'_v$. Form a new host graph from the disjoint union of these trees $T'_w$.
\item For each pair of vertices $wx$ of the host graph such that $w$ has the same color as a vertex $u$ of the pattern graph, $x$ has the same color as a vertex $v$ of the pattern graph, $u$ and $v$ are adjacent, and $w$ and $x$ are adjacent, add an edge to this union, connecting the node in $T'_w$ corresponding to the one in $T'_u$ to which the leaf for $v$ would be attached with the node in $T'_x$ corresponding to the one in $T'_v$ to which the leaf for $u$ would be attached.
\end{itemize}
The ETH-based lower bounds on the time for cubic partitioned subgraph isomorphism come from Corollary 6.1 of Marx~\cite{Mar-ToC-10}, using the fact that when $G$ is a cubic expander graph its treewidth is $\Omega(k)$~\cite{GroMar-JCTB-09}.


To translate an instance of cubic partitioned subgraph isomorphism to the problem of finding lineups in yards, we construct a yard with a row for each vertex or edge of the pattern graph $G$. We will arrange the guards and inmates of the yard in such a way that collinearities of guards correspond to vertex-edge incidences in the pattern graph, and collinearities of inmates correspond to vertex-edge incidences in the host graph. In this way, every lineup of the yard will necessarily correspond to a solution to the partitioned subgraph isomorphism instance. We take some care in the construction in order to ensure that the resulting yard uses only points with integer coordinates of small magnitude.

The rows of the yard will be placed on horizontal lines, with positive integer $y$-coordinate $y_v$ or $y_{uv}$ for each vertex $v$ or edge $uv$ of the pattern graph. It will be convenient to be able to place arbitrary integer-coordinate points on two vertex rows and have the intersection of the line through these points with an edge row automatically be an integer-coordinate point. To ensure this property, we will always choose $y_{uv}=2y_v-y_u$ (where $y_v>y_u$). However, this requires that we choose the coordinates $y_u$ and $y_v$ carefully, so that no two vertices or edges have rows with equal $y$-coordinates. We make this choice greedily, choosing the coordinates $y_v$ for the vertices in an arbitrary order.
For each vertex $v$ in this order, we choose $y_v$ to be the smallest positive integer that obeys the following constraints:
\begin{itemize}
\item $y_v$ is unequal to any $y_u$ for an earlier vertex $u$
\item $y_v$ is unequal to any $y_{uw}$ for any two adjacent earlier vertices $u$ and $w$
\item $y_v$ does not make any $y_u=y_{vx}$ for any earlier vertex $u$ and earlier neighbor $x$ of $v$
\item $y_v$ does not make $y_{uw}=y_{vx}$ for any three earlier vertices $u$, $w$, or $x$ with $u$ adjacent to $w$ and $x$ to $v$.
\end{itemize}
These conditions are each determined by at most one earlier vertex $u$, at most one of three neighbors $w$ of $u$, and at most one of three neighbors $x$ of $v$. For each combination of these vertices, each condition causes one integer to be unavailable as a choice for $y_v$.
Therefore, $O(k)$ positive integers are unavailable. Because our greedy algorithm chooses the first available integer in each case, it chooses all coordinates $y_v$ and $y_{uv}$ to have magnitude $O(k)$.  Our actual row coordinates will be scaled and translated from these values, but both of these kinds of transformations preserve the property that two integer points on vertex rows determine an integer point on an edge row.

The first points we determine in the construction are the six points of the fence. We construct a complete quadrilateral with integer coordinates, and scale it by a sufficiently large factor so that there are as many integer-coordinate horizontal lines through the fence as we need rows (given the assignment of row coordinates to vertices and edges of the pattern graph)
and so that each row intersects the region inside the fence in a segment of length proportional to the diameter of the fence.
We scale these fences and rows a second time, by a polynomially-large factor $\phi$, to ensure that the following steps can be carried out using integer coordinates for the remaining points. We leave $\phi$ unspecified for now, and will later state constraints on how large it needs to be for the construction to work.

Next, we choose three intervals on the $x$-coordinate axis: a left interval that will contain the $x$-coordinates of all left guards, a middle interval that will contain the $x$-coordinates of all inmates, and a right interval that will contain the $x$-coordinates of all right guards.
The middle interval should be chosen so that its intersection with all row lines is inside the fence
(using the fact that there is a vertical line segment inside the fence that crosses all row lines)
and the left and right intervals should intersect all row lines outside the fence, on either side of it.
inside the fence and the left and right intervals should be outside it, on either side of the fence. These intervals should be sufficiently narrow that all lines through pairs of points on corresponding intervals of different rows pass in the correct way above or below the points of the fence and the intervals of other types.  In order to prevent an inmate from being collinear with two guards from other rows than its own, we also require that there be no line that crosses a row line in the left interval, a second row line in the middle interval, and a third row line in the right interval. In order to achieve this, it will be sufficient to choose the interval widths to be smaller than the diameter of the fence by a sufficiently large factor $\rho$. With this factor, we can choose the left and right guard intervals first, arbitrarily subject to the above constraints. Given this choice, let $C$ be the set of crossing points of the rows with non-horizontal lines through potential left and right guard points.
Then $C$ is a union of a cubic number of intervals (one for each triple of rows).
Each of these intervals defining $C$ is wider than the left and right intervals by a factor of $O(k)$, because the slopes of the lines through points within the left and right intervals on different rows are bounded away from zero by $\Omega(1/k)$.
Therefore, the projection of $C$ onto the $x$-axis has length $O(k^4/\rho)$ relative to the diameter of the fence. For $\rho=\rho_0 k^4$ for a sufficiently large constant factor $\rho_0$, a constant fraction of the length within the fence will
remain disjoint from $C$, and at least one interval within this constant fraction of the length will be long enough to use as the middle interval.

After making these choices, we place the two guards on each row, with the left guard having its $x$-coordinate within the left interval and the right guard within the right interval. We choose the positions of the \emph{vertex guards} (the guards on the rows that correspond to vertices of the pattern graph), one at a time. As each vertex guard is placed, it may also determine the location of one or more \emph{edge guards} (the  guards on the rows that correspond to edges of the pattern graph).
An edge guard's location is determined when the vertex guards for its two endpoints have been placed. Whenever this happens, we place the edge guard at the point where the line through the two vertex guards crosses the edge's row.
In order to ensure that the edge guards are placed within the left and right intervals,
we place each vertex guard within the middle third of its interval.
In this way, for each edge $uv$ of the pattern graph
we will produce two collinear triples of guards: the triple of left guards for $u$, $v$, and $uv$, and the triple of right guards for $u$, $v$, and $uv$. We will position the vertex guards in sufficiently general position so that these are the only collinearities between triples of guards. As with our choice of row coordinates, we achieve this absence of extra collinearities by a greedy algorithm, in which we place one vertex guard at a time, subject to the constraint that this choice should not create any undesired collinearities between the new vertex guard, the new edge guards whose position becomes determined by the new vertex guard, and the previously determined vertex and edge guards. Each potential undesired collinearity is determined by the choice of two previous guard positions (for which there are $O(k)$ choices) and possibly one of the three neighbors of the vertex whose guard we are trying to place. This implies that, when we are placing each vertex guard, there may be up to $O(k^2)$ positions that we should not place it, because placing it in one of those positions would cause a collinearity. Recall that, when constructing the fence, we left a scaling factor $\phi$ undetermined, and that the left and right interval each contain $\Omega(k\phi/\rho)$ integer coordinates, where $\rho=O(k^4)$ is the factor by which the intervals are narrower than the diameter of the fence and where the factor of $k$ comes from the earlier scaling used to make the fence big enough to include as many integer-coordinate rows as we need.
We will choose our scaling factor $\phi$ to be large enough (at least a sufficiently large multiple of $k^5$) to ensure that the middle thirds of the  left and right intervals contain more integer coordinates than the number of unavailable positions. With this choice of $\phi$, our greedy guard-placing strategy will always be able to determine a valid integer-coordinate placement for each vertex guard point. By our previous choice of the row coordinates, the edge guard points will also have integer coordinates.

Finally, we place the inmates on each row, within the middle interval. Each edge or vertex of the host graph corresponds to exactly one inmate point, which we place on the row for the corresponding edge or vertex of the pattern graph, given by the coloring of the host graph vertex or the coloring of the endpoints of the host graph edge. We choose the locations for these inmate points in exactly the same way as we chose them for the guard points, greedily one vertex at a time, with the vertex guards placed within the middle third of the middle interval. We place each edge inmate point
on the intersection between each edge row and the line through the inmate points for its two endpoints, as soon as these two vertex inmate points have been placed.
This will cause each triple of an edge and its two endpoints in the host graph to correspond to three collinear inmate points. We require in addition that no three inmate points on different rows are collinear if they do not form such a triple, and that no two points on the same row as each other coincide. Each potential undesired collinearity is determined by the choice of two previous guard positions (for which there are $O(n^2)$ choices, as the host graph may be dense) and possibly one of up to $O(n)$ neighbors of the vertex whose guard we are trying to place. This implies that there may be $O(n^5)$ positions in the row of the new vertex guard that would cause an unwanted collinearity or coincidence. We will choose our scaling factor $\phi$ to be large enough (at least a sufficiently large multiple of $k^3n^5$) to ensure that the middle third of the middle interval in which we place the inmate points contains more integer coordinates than this number of unavailable positions. In this way, we ensure the success of our greedy inmate-placing strategy.


\begin{lemma}
\label{lem:isyard}
The translation described above constructs a yard.
\end{lemma}

\begin{proof}
The construction explicitly ensures that most properties of a yard are true: there is a fence, and rows of points, each having two guards outside the fence and inmates inside the fence. There are no non-horizontal lines of four points, and all non-horizontal lines defined by pairs of guards or pairs of inmates pass above or below the other points and through the fence in the proper manner. There are no lines of three points that include both guards and inmates and are not rows. It remains to verify two properties, not stated explicitly as part of the construction: there must be no complete quadrilateral other than the fence, and no subset of guards that form the forbidden pattern on the right of \autoref{fig:forbidden}.

If either of these forbidden patterns existed, but did not include three points from the same row of the configuration, then all of its three-point lines would have to be non-horizontal lines, composed only of guards or only of inmates. But an edge point can only participate in one such line, and both forbidden patterns include at least one line all three points of which participate in two three-point lines. So all three of the points on this line would have to be vertex points, an impossibility.

\begin{figure}[t]
\centering\includegraphics[scale=0.4]{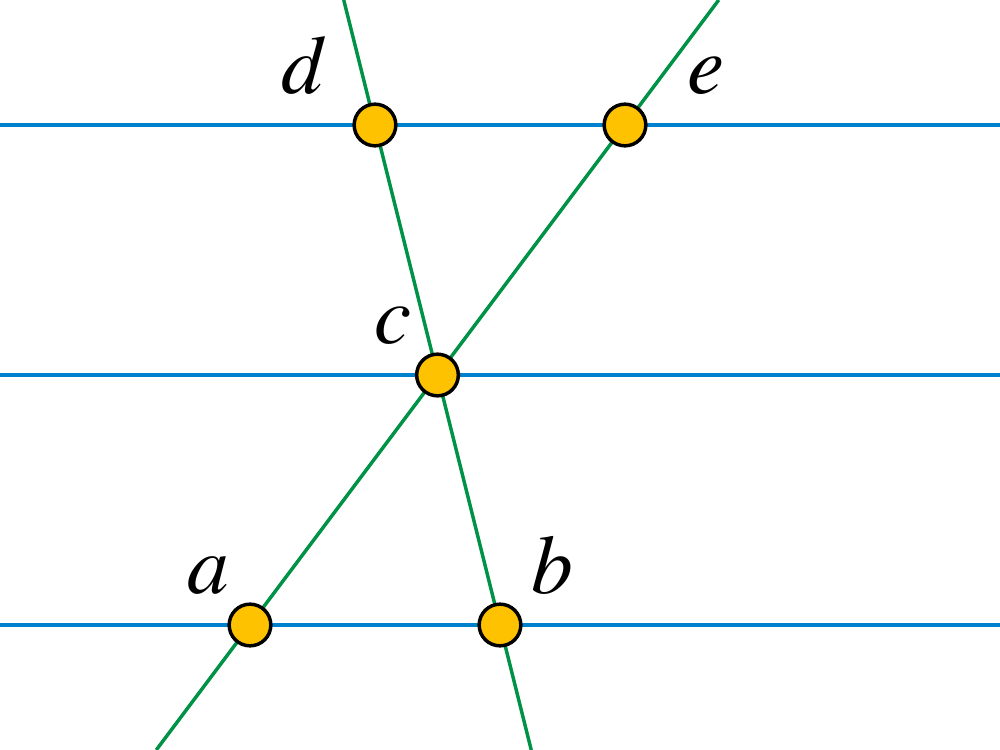}
\caption{Illustration for the proof of \autoref{lem:isyard}. The point set we construct from an instance of partitioned subgraph isomorphism cannot contain a complete quadrilateral using one row as a line $ab$, because this would force a second pair of points $de$ to belong to the same row as each other, and the parallel lines $ab$ and $de$ have no crossing point that would complete the quadrilateral.}
\label{fig:noquad}
\end{figure}

The only remaining possibility is a complete quadrilateral that uses at least one row as one of its lines. Let $a$ and $b$ be any two points of this row, $c$ be the unique point of the complete quadrilateral that belongs to three-point lines with both $a$ and $b$, $d$ be the third point of line $ac$, and $e$ be the third point of line $bc$. Then the rows containing these collinear triples of points must correspond in $G$ to an edge and its two endpoints. In particular, this implies that $d$ and $e$ are also on the same row as each other. But then lines $ab$ and $de$ are parallel, and cannot intersect to produce the sixth point of a complete quadrilateral (\autoref{fig:noquad}).
\end{proof}

\begin{lemma}
An instance $(H,G,c)$ of partitioned subgraph isomorphism (where $c$ is the coloring of the instance) is solvable if and only if the yard constructed from it admits a lineup.
\end{lemma}

\begin{proof}
If a lineup exists, it must include all fence points and guard points, and one inmate from each row.
This choice of the inmate for each row can be interpreted as selecting a vertex or edge from the host graph $H$,
of the appropriate color class (or pair of endpoint color classes), for each vertex or edge in the pattern graph $G$.
If the selected collection of vertices and edges in $H$ does not form a subgraph isomorphic to $G$, it will necessarily include two vertices $u$ and $v$ in $H$ that are adjacent in $G$ but for which the edge $uv$ does not exist in $H$, or has not been selected by the choice of inmate for its row. In this case, the guards for the rows of $u$, $v$, and $uv$ will be collinear in triples (because the corresponding vertices in $G$ are adjacent) but the inmates will not be collinear, creating an instance of the forbidden pattern on the right of \autoref{fig:forbidden}. So when $H$ does not contain a subgraph isomorphic to $G$, the yard admits no lineup.

If, on the other hand, $H$ does contain a subgraph isomorphic to $G$, we can choose the points corresponding to the vertices and edges of that subgraph as the inmates of a lineup. As the following case analysis shows, the resulting system of points does not include any forbidden pattern.
\begin{itemize}
\item It does not include four points in a line.
\item Because it has only one inmate on each row and the fence is the only complete quadrilateral in the yard, it does not include the middle forbidden pattern in \autoref{fig:forbidden}.
\item It does not include a copy of the right forbidden pattern in which none of its lines are rows, or in which only one of the vertical lines is a row, because in that case (as in \autoref{lem:isyard}) the other vertical line would be a line of three points, each of which belongs to two non-row lines of three points. This is an impossibility because the points on edge rows only belong to one non-row three-point line, and each non-row three-point line includes at least one of these edge points.
\item It does not include a copy of the right forbidden pattern in which only one or two of its horizontal lines is a row, because then the two vertical lines could only be lines through three guards. This would force the remaining horizontal lines to contain a left guard, a right guard, and a point between them, and therefore be rows themselves.
\item It does not include a copy of the right forbidden pattern with two rows as its two three-point vertical lines. If these rows represent two vertices $u$ and $v$ of $G$, the three points not on these lines would all belong to the row for edge $uv$ (as they each belong to a line containing a representative for $u$ and a representative for $v$). If the rows represent a vertex $u$ and edge $uv$ of $G$, the three points not on these lines would all belong to the row for vertex $v$. In either case these three points form another line, unlike the forbidden pattern.
\item The only remaining possibility is that the three horizontal lines of this pattern are rows. To form this pattern, the guards of these three rows must be collinear, meaning that they represent two adjacent vertices and the edge between them in $G$. But since we have selected the inmates on each row to form a correctly-colored subgraph isomorphic to $G$, the three inmates must also be collinear, preventing the forbidden pattern from occurring.
\end{itemize}
Thus, the selected points form a lineup of the yard, as required.
\end{proof}

By combining these results with the known hardness of partitioned subgraph isomorphism we obtain our main result:

\begin{theorem}
Finding a $k$-vertex compliant subset of a given $n$-point set is $\Wone$-complete and, under the exponential time hypothesis, cannot be solved in time $n^{o(k/\log k)}$.
\end{theorem}

\section{Collinearity and convex subsets}

In the introduction we claimed that finding a $k$-point subset with a given hereditary property is fixed-parameter tractable whenever the property is true for all collinear sets, or whenever there exists a convex polygon for which the property is not true. We split up the proof into three cases:
\begin{itemize}
\item Properties that are true for all collinear point sets and for all convex polygons.
\item Properties that are neither true for all collinear point sets nor for all convex polygons.
\item Properties that are true for all collinear point sets but not true for all convex polygons
\end{itemize}

\begin{figure}[t]
\centering
\includegraphics[scale=0.3]{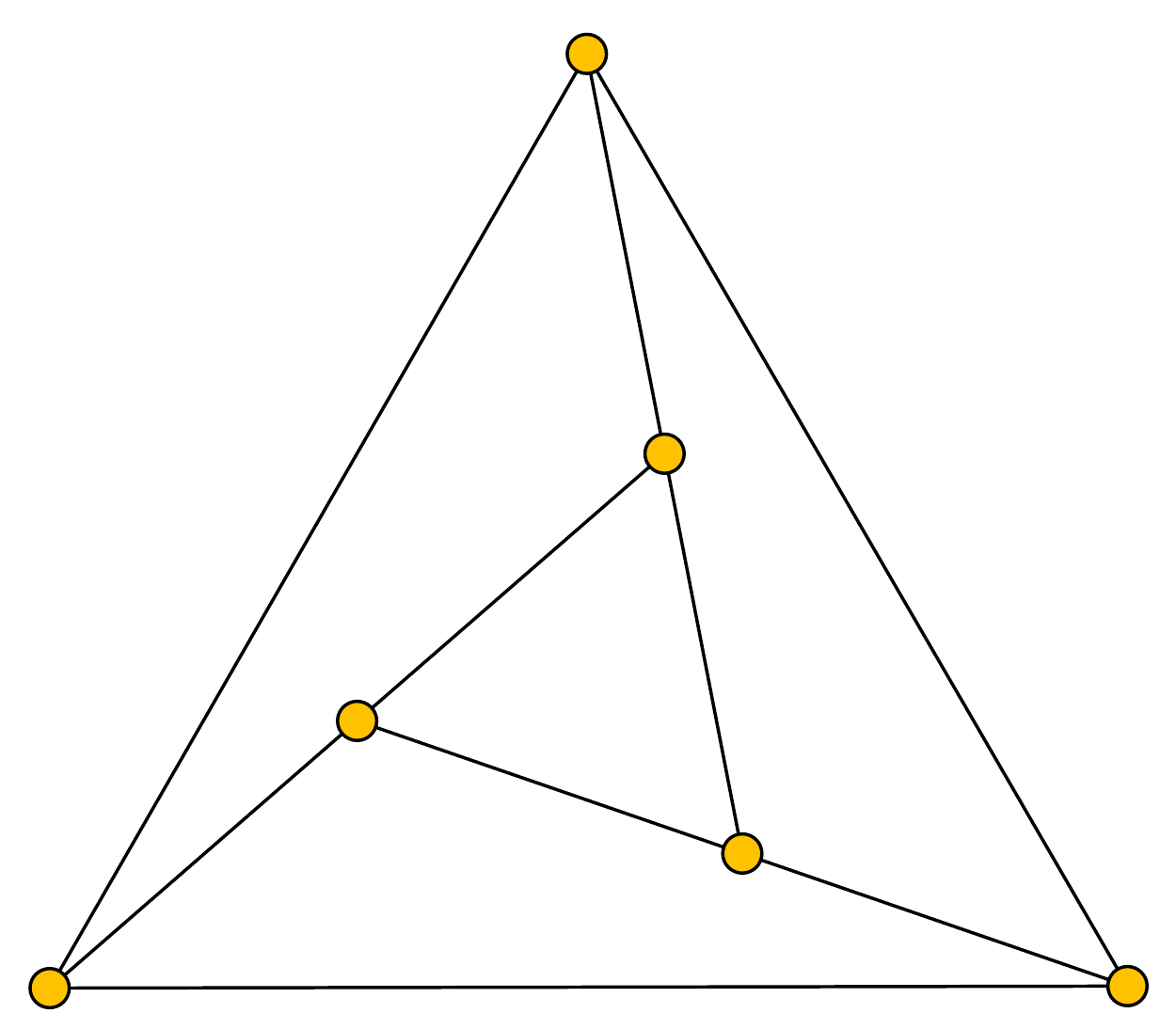}
\caption{The largest possible point set that has neither four collinear points nor four points in convex position. In general, every set of $n$ points contains $\Omega(\log n)$ points that are either collinear or in convex position. This point set forms an extreme case for that result.}
\label{fig:no4concol}
\end{figure}

However, the first two of these cases follow easily from known results:
\begin{itemize}
\item The first case covers properties that are true for all collinear subsets and for all convex polygons.
A variant of the happy ending theorem states that every set of $n$ points in the plane contains a subset of $\Omega(\log n)$ points that are either all collinear or all in convex position (\autoref{fig:no4concol}). For, if an $n$-point set has no large collinear subsets, it necessarily has a general-position subset of size nearly $\sqrt n$~\cite{PayWoo-SIDMA-13}, and applying the usual form of the happy ending theorem~\cite{ErdSze-CM-35,Suk-JAMS-17} to this subset gives a convex polygon of size $\Omega(\log n)$.
So, if property $\Pi$ is true for all collinear subsets and for all convex polygons (the first case),
then every sufficiently large point set has a $k$-point subset with property $\Pi$. Here, ``sufficiently large'' means large enough that this $\Omega(\log n)$ bound is at least $k$. For smaller values of $n$ we can search by brute force all subsets of the point set in time that is fixed-parameter tractable in $k$, because (for these values) $n$ itself is bounded by a function of $k$. And for larger values of $n$, we can either just answer yes (for a decision problem) or reduce the input arbitrarily to the smallest number of points that ensures the existence of a solution (a number bounded by a function of~$k$) and then perform a brute force search for a solution within the reduced subproblem.

\item The second case covers properties that are neither true for all collinear subsets nor for all convex polygons. Again, by the happy ending theorem, the only sets of points that can have such a property also have bounded size (independent of the parameter $k$). This follows because all sufficiently large sets
contain a collinear or convex subset without the property, and (because the property is hereditary) do not have the property themselves. In particular, if $k$ itself is sufficiently large in this sense, no solution can exist.
To test whether there exists a $k$-point subset with such a property, we compare $k$ to this bound, and answer no immediately when $k$ is too large. If $k$ is not too large, we perform a brute-force search over all $\tbinom{n}{k}$ subsets of the input of size $k$. Because we only perform this search when $k$ is bounded by a constant, the running time for this brute force search can be bounded by a polynomial with constant exponent.
\end{itemize}

We now detail a fixed-point-tractable solution for the remaining case.

\begin{lemma}
\label{lem:coll-and-noncon}
Let decidable hereditary property $\Pi$ be true for all collinear point sets but not true for all convex polygons.
Then it is fixed-parameter tractable to find a $k$-point subset having property $\Pi$, among a given set of points in the plane.
\end{lemma}

\begin{proof}
It is straightforward, in polynomial time, to find all lines through pairs of input points and determine whether any of these lines contains $k$ or more points. If so, we can immediately return a subset of $k$ collinear points. So for the remainder of the proof we assume without loss of generality that each line contains at most $k-1$ input points.

Let $q$ be the smallest number of vertices in a convex polygon that does not have property $\Pi$. By assumption, $q$ exists, and by the definition of hereditary properties, any set of points containing the vertices of a convex $q$-gon does not have property $\Pi$. Let $r$ be the largest number of points in a set of points in general position (no three in a line) that does not include a convex $q$-gon; by Suk's improvement to the happy ending theorem~\cite{Suk-JAMS-17}, $r\le 2^{q+o(q)}$.
Let $P$ be the unknown $k$-point subset (assuming it exists), and $G$ be a maximal subset of $P$ that is in general position (meaning that no other point from $P$ can be added to $G$ without creating a three-point line).
Then $|G|\le r$, or else $G$ and $P$ would contain a convex $q$-gon and $P$ would not have property $\Pi$. The pairs of points in $G$ determine $\tbinom{|G|}{2}\le\tbinom{r}{2}$ lines,
and these lines together must cover all of $P$, or else $G$ would not be maximal (\autoref{fig:maxgen}).

\begin{figure}
\centering\includegraphics[scale=0.3]{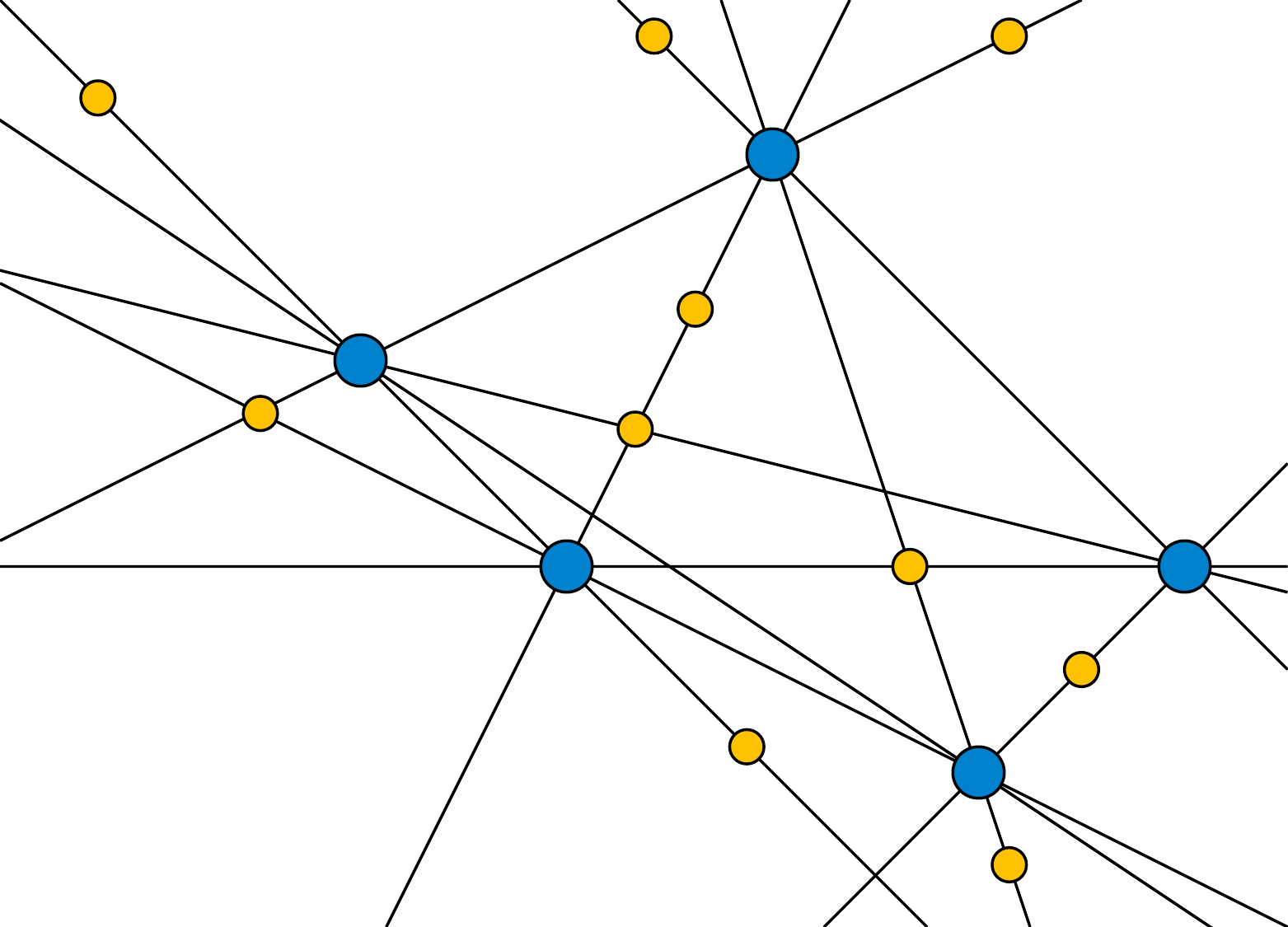}
\caption{If $G$ (blue) is a maximal general-position subset of $P$, then the remaining points of $P$ (yellow) must be covered by the lines through pairs of points in $G$, or else they could be added to $G$ to produce a larger general-position subset.}
\label{fig:maxgen}
\end{figure}

Therefore, we can solve the problem of finding a $k$-point subset $P$ with property $\Pi$ by guessing an $\tbinom{r}{2}$-tuple of lines that cover $P$, among all $\tbinom{n}{2}$ lines through pairs of input points, and then for each tuple of lines performing a brute-force search among the points covered by those lines. Because $r$ is a constant (depending on $\Pi$ but not on the input parameter $k$), there are a polynomial number of tuples of lines to be searched. Each tuple covers at most  $(k-1)\tbinom{r}{2}$ points, so
performing each brute-force search takes time bounded by a function of $k$, independent of the input size. Therefore, this algorithm is fixed-parameter tractable.
\end{proof}

This completes the last case of the following result:

\begin{theorem}
\label{thm:coll-or-noncon}
Let decidable hereditary property $\Pi$ either be true for all collinear point sets or not be true for at least one convex polygon (or both). Then it is fixed-parameter tractable to find a $k$-point subset having property $\Pi$, among a given set of points in the plane.
\end{theorem}

Not every hereditary property is decidable: there are uncountably many hereditary properties,\footnote{This follows from the fact that there are infinite families of point sets in which no family member has the same order type as a subset of another family member. For example, the set of vertices and edge midpoints of a convex $r$-gon is not a subset of the set of vertices and edge midpoints of any other convex polygon, so choosing one such point set for each $r$ gives a family of point sets none of which have the order type of a subset of another. This family has uncountably many subfamilies, each of which forms the set of forbidden patterns for a hereditary property. See \cite[Observation 5.9]{Epp-18}.} only countably many of which are decidable.  When a hereditary property is not decidable, the same technique applies, but (because of the need to test the property within each brute-force search used as a subroutine of the algorithms) it gives us a non-uniform fixed-parameter tractable algorithm.

\section{Properties with a single obstacle}

If $X$ is any fixed point set, we can define hereditary property $\Pi_X$ to be true for the point sets that do not contain $X$, and false otherwise. We are then interested in testing whether a given input includes $k$ points with property $\Pi_X$. That is, are there at least $k$ points in the largest point set that avoids $X$? If $X$ is not itself collinear, then all collinear point sets have property $\Pi_X$.
Therefore, there is only one family of possible choices for $X$ that is not already covered by \autoref{thm:coll-or-noncon}: the family of sets $X$ consisting of $q$ points on a single line. The order type of any such point set depends only on~$q$. We call a point set with this order type $L_q$.

\begin{lemma}
Let $q$ be fixed. Then it is fixed-parameter tractable to find a $k$-point subset avoiding $L_q$, among a given set of points in the plane.
\end{lemma}

\begin{proof}
We assume $q>2$ for otherwise the problem is trivial (an $L_2$-avoiding set can only be a single point, and an $L_1$-avoiding set must be empty).
We begin by finding a maximal subset $G$ of the input that is in general position (no three in a line).
If $|G|\ge k$, we return any $k$ points from $G$ as our $k$-point $L_q$-avoiding set.
Otherwise, as in \autoref{lem:coll-and-noncon} and \autoref{fig:maxgen}, the input can be covered by
at most $\tbinom{k-1}{2}$ lines, the lines through pairs of points in $G$.

Observe that, on each of these lines, an $L_q$-avoiding set can include at most $q-1$ points from the line. Consider a single line $\ell$ and a process in which we build up a $k$-point $L_q$-avoiding set, one point at a time, starting from a subset of points that are disjoint from $\ell$ and then adding points from $\ell$ in an arbitrary order. At each step of this process, at most $q-2$ points from $\ell$ have already been chosen, and the next point must be chosen to be disjoint from these already-chosen points and from any of the lines through pairs of points (neither on $\ell$) that already contain $q-1$ points. There at most $\tbinom{k-1}{2}$ of these lines through pairs of already-chosen points (because we have already chosen at most $k-1$ points), and each can prevent only a single point of $\ell$ from being chosen. So, as long as $\ell$ contains $\tbinom{k-1}{2}+q-1$ points, there will always remain at least one point that is free to be added until we have included exactly $q-1$ points from $\ell$. Therefore, if we reduce the input to contain any arbitrary subset of $\tbinom{k-1}{2}+q-1$ points on $\ell$, keeping the points disjoint from $\ell$ unchanged, we cannot affect the existence or nonexistence of a $k$-point $L_q$-avoiding subset.

By applying this observation repeatedly we can kernelize the problem. That is, we find a polynomial-time algorithm that reduces any parameterized instance to an equivalent instance (the ``kernel'') whose size is bounded by a function of the parameter. To do so, we consider the lines through pairs of points of $G$, one at a time. Whenever any such line contains more than $\tbinom{k-1}{2}+q-1$ points of the remaining input, we delete arbitrarily chosen points until it contains exactly $\tbinom{k-1}{2}+q-1$ points. At the end of this process, the total number of remaining points is at most
\[
\binom{k-1}{2}\left(\binom{k-1}{2}+q-1\right)=O(k^4),
\]
a bound depending only on $k$. We may then search for a $k$-point $L_q$-avoiding subset by a brute force search on the points of this kernel.
\end{proof}

This includes as a special case the problem of finding an $L_3$-avoiding subset (that is, a subset in general position), the central computational task for the no-3-in-line problem. (We already proved this special case to be NP-hard in \cite[Theorem 9.3]{Epp-18} and fixed-parameter tractable in \cite[Theorem 9.5]{Epp-18}.)
It completes the last case of the following result:

\begin{theorem}
Let $X$ be any fixed finite set of points. Then it is fixed-parameter tractable to find a $k$-point subset avoiding $X$, among a given set of points in the plane.
\end{theorem}

\section{Conclusions}

We have identified a hereditary property of point sets such that finding a $k$-point subset with this property that is hard to solve by fixed-parameter tractable algorithms, and identified several general classes of properties for which the corresponding problems have fixed-parameter tractable algorithms. These include properties that include all collinear sets, properties that exclude at least one convex polygon, and properties definable by a single forbidden pattern. There are several additional cases that are fixed-parameter tractable:
\begin{itemize}
\item The largest convex subset (avoiding two forbidden patterns, one of three points on a line and the other of a triangle with an interior point) is polynomially solvable, in cubic time and linear space~\cite{ChvKli-SEC-80,EdeGui-JCSS-89}.
\item If finding $k$-point subsets is fixed-parameter tractable for two properties $\Pi$ and $\Pi'$, then it is also fixed-parameter tractable for their disjunction (the point sets that have at least one of the two properties), by running both algorithms and returning any $k$-point set found by either of them.
\item If finding $k$-point subsets is fixed-parameter tractable for property $\Pi$, then it is also fixed-parameter tractable for the property $\Pi_q$ of being partitionable into $q$ subsets that have property $\Pi$, by color-coding~\cite{AloYusZwi-JACM-95}. Alternatively one could use a different hereditary property for each color class.
\end{itemize}
We have been unable to identify any pattern in the remaining cases that would enable us to distinguish them from our hard $k$-point subset problem and complete a classification of which of these problems are fixed-parameter tractable and which are not. We leave this as open for future research.

Another open problem concerns properties (such as being in convex position) that can be true only of point sets in general position. That is, these properties have a three-point line as one of their forbidden patterns. Can it be hard, for such a property, to find a $k$-point subset with the property? We suspect that the answer is yes, by a reduction that mimics the one we have given, using more complex constraints on general-position subsets of points to make them act like the collinear subsets of points in our reduction. However, we have not worked out the details of such a reduction.

\raggedright
\bibliographystyle{plainurl}
\bibliography{patterns}
\end{document}